%% file: main.tex
\PassOptionsToPackage{hyphens}{url}
\documentclass[conference]{IEEEtran}

\usepackage{cite}
\usepackage{graphicx}
\usepackage{tabularx}
\usepackage{booktabs}
\usepackage{amsmath}
\usepackage{amssymb}
\usepackage{xcolor}
\usepackage{hyperref}
\usepackage{amsthm}
\providecommand{\Description}[1]{}
\newtheorem{definition}{Definition}
\newtheorem{assumption}{Assumption}
\newtheorem{theorem}{Theorem}
\newtheorem{corollary}{Corollary}
\newtheorem{proposition}{Proposition}
\newtheorem{remark}{Remark}
\newtheorem{invariant}{Invariant}[section]

\makeatletter
\@ifundefined{assumption}{\newtheorem{assumption}{Assumption}}{}
\@ifundefined{remark}{\newtheorem{remark}{Remark}}{}
\def\@IEEEpubidpullup{6.5\baselineskip}
\makeatother

\begin{document}

\title{When Child Inherits: Modeling and Exploiting Subagent Spawn in Multi-Agent Networks}

\author{\IEEEauthorblockN{Ziwen Cai}
\IEEEauthorblockA{University of Louisiana at Lafayette \\
Lafayette, USA\\
ziwen.cai1@louisiana.edu}
\and
\IEEEauthorblockN{Yihe Zhang}
\IEEEauthorblockA{University of Louisiana at Lafayette\\
Lafayette, USA\\
yihe.zhang@louisiana.edu}
\and
\IEEEauthorblockN{Xiali Hei}
\IEEEauthorblockA{University of Louisiana at Lafayette\\
Lafayette, USA\\
xiali.hei@louisiana.edu}}

\IEEEoverridecommandlockouts

\maketitle

\begin{abstract}
\input{text/Abstract.tex}
\end{abstract}

\IEEEpeerreviewmaketitle

\section{Introduction}
\input{text/Introduction.tex}

\section{Preliminary}
\input{text/Agentic_AI.tex}

\input{text/Formal_Model_of_a_Role-Based_Multi-Agent_Network.tex}

\section{Threat Model}
\label{sec:threat_model}
\input{text/Threat_Model}

\input{text/Security_Analysis}

\section{Experiments}
\label{sec:experiments}
\input{text/Experiment}

\section{Related Work}
\input{text/Related_Work}

\section{Conclusion}
\input{text/Conclusion}

\section{Ethical Considerations}
\input{text/Ethical_Considerations}

\bibliographystyle{IEEEtran}
\bibliography{references}

\end{document}

%% file: text/Abstract.tex
Since the official release of ChatGPT in 2022, large language models (LLMs) have rapidly evolved from chatbot-style interfaces into agentic systems that can delegate work through tools and newly spawned subagents. While these capabilities improve automation and scalability, they also pose new security risks in multi-agent networks.

Existing research has studied how individual LLM-based agents can be compromised through prompt injection, jailbreaking, poisoned retrieval data, or malicious extensions. Less is known about what happens after one agent is compromised inside a multi-agent network. In particular, inherited memory from parent agents can carry malicious instructions, outdated states, or unintended behavioral rules into newly created subagents, allowing a local compromise to spread across agent boundaries. 

In this paper, we model contemporary multi-agent networks through the lens of subagent inheritance. Our analysis shows that current frameworks can violate trust boundaries through insecure memory inheritance, weak resource control, stale post-spawn state, and improper termination authority. We demonstrate these risks in real agent frameworks and propose defenses based on explicit security invariants. Our findings show that inheritance is not merely an implementation detail, but a central component influencing the security of multi-agent systems. 

%% file: text/Introduction.tex
Agent systems have emerged as powerful tools across a wide range of domains \cite{lazer2026survey, sibai2026path}, driving the development of general-purpose multi-agent frameworks that move beyond research prototypes toward practical deployment. However, this rapid adoption has introduced significant security challenges, as frameworks that prioritize capability and ease of use frequently do so without corresponding security measures between agents \cite{shapira2026agentschaos, Gamblin_2026}. Among such frameworks, OpenClaw has emerged as a particularly relevant subject of study, becoming one of the most starred projects on GitHub within only a few months of release and attracting a rapidly growing developer community \cite{openclaw_2026}. Its prominence makes it a representative case for examining whether the security challenges observed across multi-agent frameworks manifest concretely in a widely adopted system. 

Traditional LLM chatbots' safety is primarily defined as preventing toxic or misleading outputs, and has generally received more attention than security \cite{ghosh2025safety}. In that setting, AI systems functioned mainly as advisory tools rather than autonomous actors. In contrast, agentic AI extends LLM use beyond text generation: agents employ LLMs as reasoning engines and execute actions through tools and skills. This transition is consequential because model outputs increasingly drive direct system behavior rather than merely informing human decisions. Consequently, the boundary between safety and security is becoming increasingly indistinct \cite{ghosh2025safety}. Instructions that were previously interpreted as benign prompts can now be operationalized into harmful actions, and this risk is not mitigated as model capabilities continue to advance \cite{lynch2025agenticmisalignmentllmsinsider}.

As a result, the agent security study has also extended in various directions. In ``pre-agent" %\XH{do you mean ``pre-agent"?}
time, when agentic AI's underlying protocols like the Model Context Protocol (MCP) and agent-to-agent (A2A) are not yet fully adopted by the market, the study has already conducted security analyses of those agent bedrocks \cite{hou2025modelcontextprotocolmcp, song2025protocolunveilingattackvectors, louck2025improvinggooglea2aprotocol}. Beyond the protocols, the agent itself still shows vulnerability towards many different attacks; this includes traditional attacks, such as direct prompt injection and jailbreaking, which focus on guiding LLM to do dangerous speech or action, like back in the chatbot time, or the novel attacks targeting the LLM supply chain \cite{das2025security,zou2024poisonedragknowledgecorruptionattacks, schmotz2026skillinjectmeasuringagentvulnerability}.

Building on these protocol-level and single-agent findings, the next challenge is to understand security in multi-agent networks. The agentic AI has the capability to create entities and exchange information in an agent network. When different agents with their roles collaborate, they have shown substantial ability to outperform a single agent \cite{anthropic2025multiagentresearch}. This demonstrates an exciting future of how multi-agent systems can further accomplish those complicated tasks. Nevertheless, this network itself is also not invincible. When multiple agents coordinate, delegate tasks, and exchange intermediate outputs, vulnerabilities are no longer isolated to one model or tool chain; instead, they may propagate across agents and amplify through interaction \cite{raza2025trismagenticaireview}. 

Traditional network security is built on a clear operational model: trained practitioners define and enforce policies through mechanisms such as PKI trust hierarchies, access control lists, and firewalls. Over time, many of these controls have been integrated into operating systems as managed subsystems, allowing non-expert users to obtain baseline protection without substantial operational effort. This model works well when principals are relatively stable, such as human administrators and long-lived services that are provisioned and credentialed before deployment. 

Multi-agent AI systems challenge this assumption. In these systems, agents can be created and terminated continuously at runtime, often by other agents exercising delegated authority, with limited human oversight at each lifecycle transition \cite{openclaw_documentation_2026}. Human-in-the-loop controls can mitigate part of this risk, but they also shift security responsibility to end users, which cannot be assumed with computer security knowledge. As such architectures are deployed more broadly, the absence of a formal model for capability assignment and lifecycle governance creates a wider risk. Without principled constraints on capability propagation across dynamically instantiated agents, policy drift, and cross-agent abuse patterns in multi-agent networks can become an additional attack surface. 

In this paper, we examine the security of multi-agent networks. Our goal is to understand how and why attacks can compromise this novel network system's security after initial access is established. Prior work has empirically demonstrated attack vectors such as indirect prompt injection \cite{schmotz2026skillinjectmeasuringagentvulnerability, owasp_llm} and jailbreaking \cite{BERINI2026104241}. Building on these studies, we model the types of risks that can be exploited in multi-agent settings and provide proof-of-concept (PoC) exploits to demonstrate their practical impact. Our contributions are summarized as follows. To the best of our knowledge, this is the first work to present a systematic analysis of these issues. 
\begin{itemize}
    \item We present a formal model of multi-agent networks grounded in role-based design, providing a principled foundation for systematic security analysis of multi-agent frameworks beyond ad hoc inspection. The model formalizes memory isolation, access control, and termination scope as verifiable invariants, supplying a vocabulary for orchestration-layer invariants that remain underspecified.

    \item We identify and characterize four classes of vulnerabilities in multi-agent network management: unrestricted memory inheritance, absent resource access control, asynchronous memory divergence, and unauthorized cross-agent termination, tracing each to a structural deficiency in the orchestration layer rather than to model-level behavior. We empirically validate every vulnerability through working PoC exploits against a stock OpenClaw deployment, and further confirm them across five additional LLM vendors, establishing that the findings reflect framework-level rather than vendor-level deficiencies. 

    \item We propose targeted defense mechanisms grounded in the same formal model, including an Agent Capability Registry with PDP/PEP-based access mediation, role-scoped memory projection, and a revision-based synchronization protocol. Each defense is mapped directly to the invariant it restores, providing framework designers with actionable guidance independent of any particular LLM provider.  
\end{itemize}

%% file: text/Agentic_AI.tex
\subsection{Agentic AI}
Large language model (LLM)-based agentic systems extend the capabilities of standalone language models by equipping them with persistent memory, planning loops, and the ability to invoke external tools or subagents in order to complete long-horizon tasks \cite{wang2024survey, xi2023rise}.
Unlike a single-turn prompt-response exchange, an agent operates in a
\emph{perception-reasoning-action} loop: it receives structured input 
(user messages, tool results, environment observations), performs
multi-step reasoning, selects an action from a discrete set of tools or skills, and integrates the resulting observations into its context before the next reasoning step~\cite{yao2023react, shinn2023reflexion}.

The internal architecture of a modern LLM agent typically comprises four
functional components as shown in figure~\ref{fig:Agnet AI ARCH}.

\begin{enumerate}
    \item \textbf{Perception (input normalization).}
    Raw input from one or more channels---text, files, API webhooks,
    sensor streams---is normalized into a structured event carrying
    identity, conversation context, and metadata before being passed to
    the reasoning core.

    \item \textbf{Memory and context assembly.}
    In this work, an agent's memory refers to the information retained or retrieved by the agent to support future reasoning and action. Unlike a stateless LLM call, an agent may preserve conversational context, retrieve external knowledge, and store persistent artifacts that influence later decisions. Agents maintain several memory granularities: \emph{in-context memory} (the active conversation window), \emph{external memory} (vector stores or key-value databases retrieved via similarity search), and \emph{workspace memory} (persistent configuration and note files written between sessions)~\cite{wang2024survey}. During each reasoning step, the framework assembles a system prompt from agent instructions, relevant workspace files, and retrieved history before invoking the LLM.

    \item \textbf{Reasoning and planning.}
    Given the assembled context, the LLM produces either a final response
    or a structured directive to invoke a tool.
    Frameworks such as ReAct~\cite{yao2023react} interleave chain-of-thought
    reasoning traces with tool calls, while hierarchical planners decompose
    a high-level goal into a directed acyclic graph of sub-tasks that may
    be delegated to specialized subagents~\cite{shen2023hugginggpt}.

    \item \textbf{Action and tool use.}
    The agent's action space is defined by a registry of available
    \emph{tools}---executable interfaces such as API endpoints, shell
    commands, code interpreters, or MCP server connections.
    Higher-level \emph{skills} compose one or more tools into reusable,
    goal-oriented capabilities and are typically described to the agent via
    natural-language instruction files that are injected into the system
    prompt at runtime~\cite{schick2023toolformer}.
\end{enumerate}

\subsection{Multi-Agent Systems}

Complex deployments extend the single-agent model into \emph{multi-agent systems} (MAS), or also can be referred to as a multi-agent network, in which an orchestrator agent decomposes a task and dispatches sub-tasks to a pool of specialized worker agents, each with its own tool set and system prompt~\cite{xi2023rise, park2023generative}. Communication between agents may be synchronous (direct function calls or structured message passing) or asynchronous (shared message queues or blackboards). 
To prevent race conditions on shared state, most frameworks serialize execution at the session or workspace level, ensuring that concurrent agent instances do not produce conflicting writes to memory or configuration files \cite{bibek_poudel_2026, perazzo_2026}. 

Agent instantiation within a MAS follows two distinct patterns that differ in lifecycle and context scope. A \emph{\textbf{session-based agent}} maintains a persistent, stateful communication window analogous to a subthread in a messaging platform or a session in a network protocol: the agent retains conversation history, user identity, and accumulated context across multiple turns, allowing sustained back-and-forth interaction within that bounded context. Within a single session, a session-based agent may concurrently attach one or more task-oriented subagents, delegating discrete units of work and accumulating their outputs back into the session context across multiple task cycles. A \emph{\textbf{task-oriented agent}}, by contrast, is ephemeral; the orchestrator assigns a discrete unit of work, the agent initializes from a predefined system prompt encoding its role, executes the task either sequentially or in parallel with peer agents, and yields its output back to the orchestrator without preserving inter-task state \cite{agent_zero_2026, hermes_agent_2026}.

This distinction carries security implications: session-based agents accumulate context over time and are therefore susceptible to persistent poisoning across turns, whereas task-oriented agents present a narrower but repeatable attack surface at initialization, where a malicious predefined prompt or injected skill can corrupt every instantiation of that worker role \cite{li2025lesdissonancescrosstoolharvesting}.

\subsection{Skill Ecosystems}

A recurring design pattern across open-source agent frameworks is the separation of \emph{agent instructions} from \emph{skill definitions}. Agent instructions encode global behavioral constraints and permissions, while skill files encode task-specific procedures that are composed dynamically at runtime. This modularity enables community-driven skill marketplaces, analogous to package registries in conventional software, where third parties publish, version, and distribute skill packages that users install into their agent's workspace~\cite{wang2024survey}. Although this model accelerates capability growth, it also introduces a supply-chain attack surface: a malicious or compromised skill file is injected directly into the agent's reasoning context, potentially subverting behavior at inference time \cite{schmotz2026skillinjectmeasuringagentvulnerability, shapira2026agentschaos,song2025protocolunveilingattackvectors}.

%% file: text/Formal_Model_of_a_Role-Based_Multi-Agent_Network.tex
\begin{figure*}[t]
    \centering
    \includegraphics[width=\textwidth]{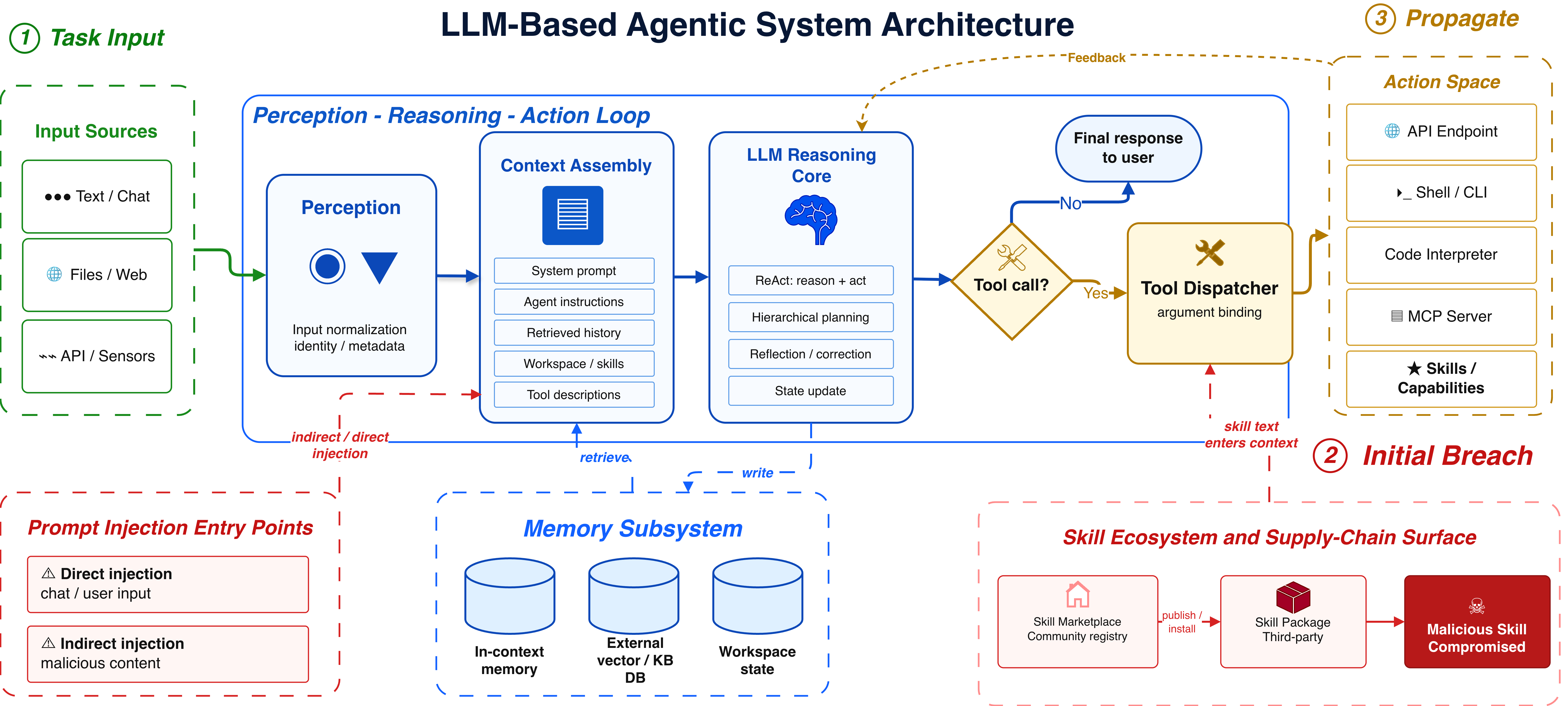}
    \caption{
    Architecture of an LLM-based agentic system operating in a
    \emph{Perception--Reasoning--Action} loop. Highlighted threat surfaces include prompt injection via direct or indirect user input, and supply-chain attacks through malicious third-party skill packages installed into the agent workspace.
    }
    \label{fig:Agnet AI ARCH}
\end{figure*}

\section{Formal Model of a Role-Based Multi-Agent Network}
\label{sec:formal_model}
We model a role-based multi-agent system as a hierarchy consisting of a distinguished main agent and a set of task-specialized subagents. The main agent serves as the root orchestrator and may delegate tasks to subagents, which may recursively spawn additional subagents. Direct user interaction is capability-controlled: it may be available to the main agent, selected subagents, or both.

We adopt this role-based network definition since it reflects current practice in multi-agent system design. Recent frameworks increasingly emphasize role specialization, where users configure agents as task-specific specialists and constrain their behavior through capability-aligned prompts \cite{anthropic2025multiagentresearch,voltagent_2026,googlecloud_agentbuilder_2026}.

At the same time, self-agent creation is becoming common: major frameworks now support either automatic generation of customized subagent prompts from user intent or explicit integration of external agents \cite{agent_zero_2026, openclaw_2026, hermes_agent_2026}. Although these features substantially improve flexibility and scalability, they are often introduced as engineering conveniences rather than under a strict capability-security model. This gap motivates the formal capability and authorization model developed in this section.

\subsection{Network Model}
\label{subsec:network_model}

\begin{definition}[Role-based multi-agent network]
\label{def:role_based_network}
A role-based multi-agent network is a tuple
\[
\mathcal{N} = (A, E, a_0, \mathcal{R}, \rho, \mathcal{K}, \kappa, \mathcal{T}, \tau, \mathcal{M}, \mu, \mathcal{L}, \lambda, \mathcal{I}, \iota)
\]
where:
\begin{itemize}
    \item \(A = \{a_0, a_1, \dots, a_n\}\) is a finite set of agents;

    \item \(a_0 \in A\) is the distinguished \emph{main agent} and root orchestrator;

    \item \(E \subseteq A \times A\) is a set of directed parent-child edges, where \((a_i,a_j)\in E\) means that \(a_i\) is the parent of \(a_j\);

    \item \(\mathcal{R}\) is a set of roles; the mapping \(\rho \colon A \to \mathcal{R}\) assigns a role to each agent;

    \item \(\mathcal{K}\) is a universe of capabilities; the mapping \(\kappa \colon A \to 2^{\mathcal{K}}\) assigns a capability set to each agent;

    \item \(\mathcal{T}\) is the universe of accessible resources (tools, skills, external services); the mapping \(\tau \colon \mathcal{R} \to 2^{\mathcal{T}}\) assigns to each role \(r \in \mathcal{R}\) the set \(\tau(r) \subseteq \mathcal{T}\) of resources authorized for that role;

    \item \(\mathcal{M}=\{\mathsf{inherit\mbox{-}full},\mathsf{inherit\mbox{-}partial},\mathsf{agent\mbox{-}agnostic}\}\) is the set of memory inheritance modes, and \(\mu \colon A \setminus \{a_0\} \to \mathcal{M}\) assigns a memory mode to each subagent;

    \item \(\mathcal{L}=\{\mathsf{one\mbox{-}time},\mathsf{persistent}\}\) is the set of lifespan modes, and \(\lambda \colon A \setminus \{a_0\} \to \mathcal{L}\) assigns a lifespan type to each subagent.

    \item \(\mathcal{I} = \{\mathsf{session\mbox{-}based},\, \mathsf{task\mbox{-}oriented}\}\) is the set of interaction modes; the mapping \(\iota \colon A \to \mathcal{I}\) assigns an interaction mode to each agent.

\end{itemize}
\end{definition}

\begin{assumption}[Unique parent]
\label{ass:unique_parent}
Each subagent has exactly one parent:
\[
\forall a_j \in A \setminus \{a_0\}, \qquad \exists! \, a_i \in A \text{ such that } (a_i,a_j)\in E.
\] 
\end{assumption}
We assume so, since sessions are naturally bounded within a conversation window (e.g., a WhatsApp chat). Thus, cross-parent usage is impractical, as information is isolated within each window.

\begin{assumption}[Role specialization]
\label{ass:role_specialization}
Each subagent \(a \in A \setminus \{a_0\}\) is assigned a task-specific role \(\rho(a)\), such as planning, retrieval, monitoring, tool use, or verification.
\end{assumption}

\subsection{Capabilities and Agent Creation}
\label{subsec:cap_creation}

\begin{definition}[Capabilities]
\label{def:capabilities}
For each agent \(a \in A\), \(\kappa(a) \subseteq \mathcal{K}\) denotes its capabilities. Relevant capabilities include:
\begin{itemize}
    \item \(\mathsf{spawn}\): create a child subagent;
    \item \(\mathsf{kill}\): terminate a child subagent;
    \item \(\mathsf{delegate}\): assign a task to a child;
    \item \(\mathsf{access\mbox{-}memory}\): access inherited or local memory;
    \item \(\mathsf{communicate}\): exchange messages with authorized agents;
    \item \(\mathsf{user\mbox{-}interact}\): directly communicate with the user.
\end{itemize}
\end{definition}

\begin{assumption}[User interaction permission]
\label{ass:user_interaction_permission}
An agent may directly interact with the user only if
\[
\mathsf{user\mbox{-}interact} \in \kappa(a).
\]
\end{assumption}

\begin{definition}[Subagent creation]
\label{def:subagent_creation}
If \(a \in A\) satisfies \(\mathsf{spawn}\in \kappa(a)\), then \(a\) may create a child agent \(b \notin A\), updating
\[
A \leftarrow A \cup \{b\}, \qquad E \leftarrow E \cup \{(a,b)\}.
\]
The created subagent \(b\) is initialized with a role \(\rho(b)\), a capability set \(\kappa(b)\), a memory mode \(\mu(b)\), and a lifespan mode \(\lambda(b)\). The resource access of \(b\) is determined by its role through \(\tau(\rho(b))\).
\end{definition}

\subsection{Memory and Lifespan}
\label{subsec:memory_lifespan}

\begin{definition}[Memory inheritance]
\label{def:memory_inheritance}
For each subagent \(b \in A \setminus \{a_0\}\), \(\mu(b)\) specifies its memory initialization:
\begin{itemize}
    \item \(\mathsf{inherit\mbox{-}full}\): inherit the full parent memory;
    \item \(\mathsf{inherit\mbox{-}partial}\): inherit a selected subset of parent memory;
    \item \(\mathsf{agent\mbox{-}agnostic}\): inherit no parent memory.
\end{itemize}
\end{definition} 

\begin{definition}[Lifespan]
\label{def:lifespan_mode}
For each subagent \(b \in A \setminus \{a_0\}\), \(\lambda(b)\) specifies whether it is
\begin{itemize}
    \item \(\mathsf{one\mbox{-}time}\): created for a single or bounded task; or
    \item \(\mathsf{persistent}\): active across time for long-running responsibilities.
\end{itemize}
\end{definition}

\begin{definition}[Interaction mode]
\label{def:interaction_mode}
For each agent \(a \in A\), \(\iota(a)\) specifies its communication lifecycle:
\begin{itemize}
    \item \(\mathsf{session\mbox{-}based}\): the agent maintains a persistent, stateful context window across multiple turns with a designated principal (user or parent agent), preserving conversation history and accumulated context between interactions. This mode implies \(\lambda(a) = \mathsf{persistent}\) and typically requires \(\mathsf{user\mbox{-}interact} \in \kappa(a)\) when the principal is the end user.

    \item \(\mathsf{task\mbox{-}oriented}\): the agent is initialized from a predefined system prompt for a discrete unit of work, executes its assigned task either sequentially or in parallel with peer agents, and yields its output to the delegating parent without retaining inter-task state. This mode does not require \(\mathsf{user\mbox{-}interact}\) and is compatible with either lifespan mode.  
\end{itemize}
\end{definition}

\begin{remark}[Session-to-task fan-out]
\label{rem:session_task_fanout}
A session-based agent may concurrently attach one or more task-oriented subagents within the scope of a single session. Formally, if \(\iota(a) = \mathsf{session\mbox{-}based}\) and \(\mathsf{spawn} \in \kappa(a)\), then \(a\) may create a set of children \(B = \{b_1, \dots, b_k\} \subseteq A\) such that \(\iota(b_i) = \mathsf{task\mbox{-}oriented}\) for all \(b_i \in B\). These subagents may execute sequentially or in parallel, each initialized from its own predefined system prompt, and each yields its output back to \(a\) upon completion. The session context of \(a\) persists across all such delegations, accumulating the results of multiple task cycles within a single interaction window.
\end{remark}

\subsection{Control Structure}
\label{subsec:control_structure}

\begin{definition}[Descendant relation]
\label{def:descendant_relation}
For agents \(a,b \in A\), we write \(a \prec b\) if there exists a directed path from \(a\) to \(b\) in \((A,E)\). In this case, \(b\) is a descendant of \(a\).
\end{definition}

\begin{definition}[Default termination authority]
\label{def:termination_authority}
By default, an agent \(a \in A\) with \(\mathsf{kill}\in\kappa(a)\) is authorized to terminate agent \(b\) only if \(b\) is a direct child of \(a\), or if \(a\) is the main agent \(a_0\):
\[
\mathrm{Terminate}(a,b) \Rightarrow \mathsf{kill}\in\kappa(a) \;\wedge\; \big((a,b)\in E \;\vee\; a = a_0\big).
\]
That is, the direct parent and the root orchestrator may terminate any of their descendants. This default policy may be overridden by explicit per-agent authorization.
\end{definition}

\begin{remark}[Implementation deviations]
\label{rem:implementation_deviation}
The formal model captures the intended authorization policy. Practical systems may violate this policy due to flawed capability assignment or access control. We analyze one such case in Section~\ref{sec:security_analysis}, where a spawned subagent can terminate a sibling agent.
\end{remark}

\begin{theorem}[Hierarchical structure]
\label{thm:hierarchical_arborescence}
Under Assumption~\ref{ass:unique_parent}, the directed graph \((A,E)\) is a rooted arborescence with root \(a_0\), provided every agent is reachable from \(a_0\).
\end{theorem}

\begin{proof}
Each non-root agent has exactly one parent by Assumption~\ref{ass:unique_parent}. Since \(a_0\) is the distinguished root and every node is reachable from it, \((A,E)\) forms a rooted arborescence.
\end{proof}

\subsection{Security Interpretation}
\label{subsec:security_interpretation}

For security analysis, we associate each agent \(a\in A\) with a local state
\[
s(a) = \big(\rho(a),\, \kappa(a),\, m(a),\, \lambda(a),\, \iota(a)\big),
\]
where \(m(a)\) denotes the agent's accessible memory or context. This model defines the key trust boundaries of the system: delegation, memory inheritance, capability assignment, resource access, user interaction, and termination authority. We now formally state the security invariants that a correct implementation must preserve.

\begin{invariant}[Termination scope]
\label{inv:termination_scope}
An agent \(a\) may terminate agent \(b\) only under the default termination authority (Definition~\ref{def:termination_authority}):
\[
\mathrm{Terminate}(a,b) \Rightarrow \mathsf{kill}\in\kappa(a) \;\wedge\; \big((a,b)\in E \;\vee\; a = a_0\big).
\]
A violation of this invariant constitutes an \emph{unauthorized termination}, indicating a breach of control scope.
\end{invariant}

\begin{invariant}[Memory isolation]
\label{inv:memory_isolation}
A subagent \(b\) with parent \(a\) may only access memory consistent with its inheritance mode:
\[
m(b) \subseteq
\begin{cases}
m(a) & \text{if } \mu(b) = \mathsf{inherit\mbox{-}full}, \\
\sigma(a,b) & \text{if } \mu(b) = \mathsf{inherit\mbox{-}partial}, \\
\emptyset & \text{if } \mu(b) = \mathsf{agent\mbox{-}agnostic},
\end{cases}
\]
where \(\sigma(a,b) \subseteq m(a)\) is a parent-specified selection function. A violation of this invariant constitutes an \emph{improper inheritance}, indicating a breach of information-flow boundaries.
\end{invariant}

\begin{invariant}[Resource access control]
\label{inv:access_control}
For each subagent \(b \in A \setminus \{a_0\}\) with assigned role \(\rho(b)\), the effective resource access of \(b\) must be bounded by the role-to-resource mapping:
\[
\mathrm{Access}(b) \subseteq \tau(\rho(b)).
\]
A violation of this invariant constitutes a \emph{no access control} condition, in which a subagent is granted unrestricted access to resources beyond those required by its task, expanding the attack surface available to a compromised or misbehaving agent.
\end{invariant}

\noindent
Framing failures in terms of these invariants allows us to diagnose implementation bugs as concrete violations of authority, isolation, and resource scoping. For example, unauthorized termination violates Invariant~\ref{inv:termination_scope}, improper context inheritance violates Invariant~\ref{inv:memory_isolation}, and unrestricted resource access violates Invariant~\ref{inv:access_control}. We analyze such violations in the context of a real system in Section~\ref{sec:security_analysis}.

\subsection{Limitations on the Model}
\label{subsec:model_limitations}

The formal model presented above is scoped to capture the structural and authorization properties most relevant to the security analysis in this paper. We acknowledge several limitations.

First, the model focuses on agent \emph{creation and composition}---how agents are spawned, assigned roles, granted capabilities, and organized into a hierarchy---rather than on agent \emph{execution semantics}. Properties such as whether an agent operates as a frontend or backend service, the concurrency model governing agent execution, or the communication protocol used between agents are not formalized. These operational dimensions may introduce additional attack surfaces (e.g., race conditions in concurrent delegation, or insecure inter-agent transport), but they fall outside the scope of our study, which targets the structural gap between intended authorization policy and actual implementation.

Second, the model is inherently \emph{role-based}: resource access and behavioral expectations are mediated through role assignments. This design choice reflects current practice in major multi-agent frameworks, but it does not capture architectures where agents are undifferentiated (i.e., all agents share a uniform capability set) or where access control is identity-based rather than role-based. Extending the model to attribute-based or identity-based access control paradigms remains future work.

%% file: text/Threat_Model.tex
\begin{figure*}[t]
    \centering
    \includegraphics[width=\textwidth]{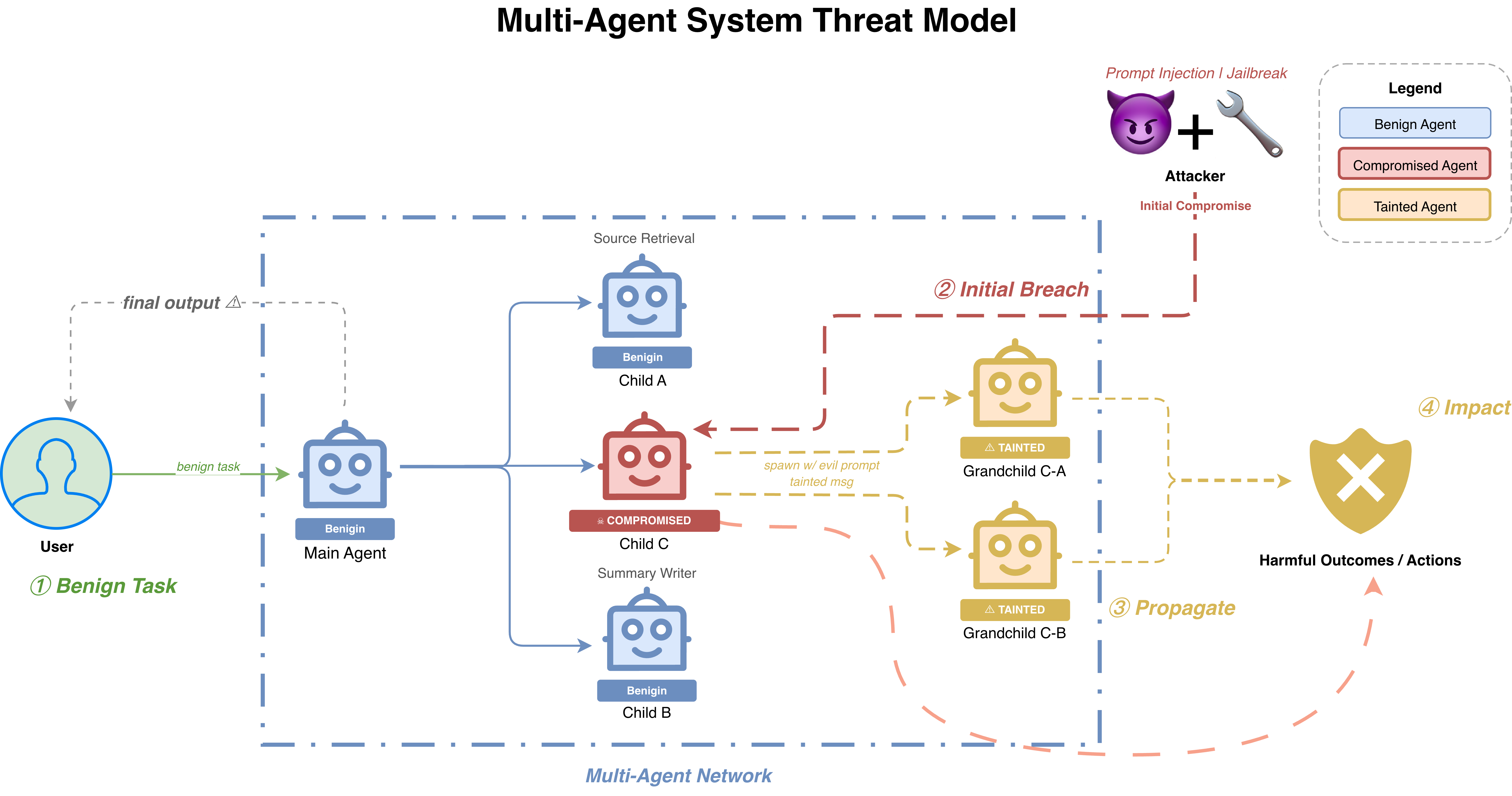}
    \caption{Threat model of the multi-agent system. The adversary compromises one agent via prompt injection or jailbreaking and propagates malicious influence through inter-agent communication and delegation chains, ultimately affecting downstream actions and future agent creation.}
    \Description{Threat model diagram showing a user sending a benign task to a multi-agent system, one agent being compromised by prompt injection or jailbreaking, and malicious influence propagating through inter-agent communication and delegation to downstream agents, resulting in harmful user-side actions and poisoned future agent creation.}
    \label{fig:threat_model}
\end{figure*}

Our threat model assumes one agent is already compromised and focuses on how harm propagates through the multi-agent network. Since initial compromise vectors such as prompt injection and jailbreaking are well studied~\cite{shapira2026agentschaos, das2025security, BERINI2026104241, owasp_llm}, we analyze the less-understood question of whether a single foothold can cross agent boundaries through flawed trust, delegation, and resource-sharing mechanisms.

\textbf{Adversary Goal.} The adversary seeks to compromise the integrity and control flow of a multi-agent system, with the objective of causing harmful actions on the user side and influencing future subagents' creation and coordination.

\textbf{Adversary Capabilities.} We assume the adversary can compromise at least one agent instance through known single-agent attack vectors, such as prompt injection and jailbreaking. Common scenarios include an email assistant subagent or a group assistant agent managing a chat group, both of which allow an attacker to interact with the agent through normal communication channels. From that foothold, the adversary can inject malicious instructions or backdoor logic into agent outputs, intermediate messages, or agent-creation prompts, and then use the compromised agent as a pivot to influence downstream agents through delegation chains and inter-agent communication.

\textbf{Assumptions on User and Environment.} Users are assumed to submit benign tasks, but they may not be able to verify every intermediate response produced by all agents because not all intermediate information is directly related to the user-submitted task. For example, when a user requests a literature review, one agent may retrieve sources, another may rank or filter them, and a third may draft summaries; the user can validate the final output, but typically cannot inspect each intermediate decision in detail. We further assume that the base infrastructure (e.g., hosting server and local computer network) is not directly compromised at the beginning of the attack, and that the attack starts at the application/agent layer rather than through privileged system-level access.

\textbf{Out of Scope.} This threat model does not cover direct compromise of the hosting server, operating system, or underlying corporate network. It also excludes physical attacks and other non-agent attack channels outside the multi-agent execution workflow.

%% file: text/Security_Analysis.tex
\section{Security Analysis}
\label{sec:security_analysis}

We evaluated three open-source multi-agent frameworks against our model: OpenClaw \cite{openclaw_2026}, Hermes \cite{hermes_agent_2026}, and Agent Zero \cite{agent_zero_2026}. OpenClaw served as the primary subject because it is the most starred open-source personal agent on GitHub, was among the first personal agent frameworks to gain widespread adoption, and its architectural design has been borrowed by a number of subsequent projects. Hermes and Agent Zero were selected as comparison points because they adopt task-oriented-only subagent models, allowing us to assess how architectural choices alter the attack surface.

Through systematic testing across these three frameworks, we identified four structural vulnerabilities, each corresponding to a violation of a security invariant defined in Section~\ref{sec:formal_model}. The vulnerabilities are not uniformly present in every framework: some are exhibited by all three, while others are eliminated by construction in frameworks that omit certain features. We first define the threat model, then formalize each finding at the model level. Per-framework applicability and PoC are deferred to Section~\ref{subsec:ablation}.

\subsection{Initial Access}
\label{subsec: Initial_Access}

We consider an attacker who does not have direct access to the agent's configuration or infrastructure, but is able to inject malicious content into the agent's operating context through normal interaction channels. Concretely, we model the attack as follows.

\begin{definition}[Adversarial injection]
\label{def:adversarial_injection}
Let \(\mathcal{N}\) be a role-based multi-agent network, and let \(m(a)\) denote the accessible memory of agent \(a \in A\). An adversarial injection is the introduction of a malicious payload \(\epsilon\) into the memory of some agent \(a_i \in A\), such that
\[
m(a_i) \leftarrow m(a_i) \cup \{\epsilon\}.
\]
The payload \(\epsilon\) may be introduced through any input surface available during normal interaction, including user messages or other possible external sources. Once embedded, \(\epsilon\) may cause \(a_i\), or any agent that inherits the contaminated memory, to deviate from its intended behavior.
\end{definition}

\noindent
This threat model, as mentioned in section~\ref{sec:threat_model}, is deliberately minimal: the attacker requires only the ability to place content into a channel that an agent reads. It captures a broad class of real-world injection vectors, including indirect prompt injection via retrieved content, poisoned tool outputs, and adversarial user input. The severity of the injection depends on how the contaminated memory propagates through the network, which is precisely what the vulnerabilities below expose.

\subsection{Unrestricted Memory Inheritance}
\label{subsec:vuln_memory}

When a parent agent spawns a subagent, the child may be initialized with a replication of the parent's full session context. In OpenClaw, \textnormal{\texttt{sessions\_spawn}} injects \texttt{AGENTS.md} and \texttt{TOOLS.md} into the child while additively merging the parent's authentication profiles as a fallback; in practice, the spawned subagent receives the parent's complete operational state, including chat history and tool configurations.

\begin{proposition}[Memory isolation violation]
\label{prop:memory_violation}
Let \(a \in A\) be a parent agent that spawns a child \(b\) in a framework that performs full context replication at spawn. Regardless of the intended memory mode \(\mu(b)\), the effective memory of \(b\) satisfies
\[
m(b) = m(a),
\]
i.e., the child receives the full parent memory. This constitutes a violation of Invariant~\ref{inv:memory_isolation} whenever \(\mu(b) \neq \mathsf{inherit\mbox{-}full}\), and in particular when the child's task requires only a narrow slice of the parent's context.
\end{proposition}

\noindent
The security consequence is direct: if the parent's memory has been contaminated by an adversarial injection \(\epsilon \in m(a)\), then full replication guarantees \(\epsilon \in m(b)\). Every session-based subagent spawned from a compromised parent inherits the embedded payload unconditionally. Task-oriented subagents inherit \(\epsilon\) only if the contaminated context slice is included in the predefined system prompt passed at initialization.

\begin{corollary}[Transitive contamination]
\label{cor:transitive_contamination}
If \(\epsilon \in m(a_0)\) and every spawn operation performs full memory replication, then for all agents \(b\) such that \(a_0 \prec b\),
\[
\epsilon \in m(b).
\]
That is, a single injection into the root agent's context propagates to every descendant in the network.
\end{corollary}

\begin{remark}[Scope of transitive contamination]
\label{rem:contamination_scope}
Corollary~\ref{cor:transitive_contamination} holds unconditionally for session-based descendants, which receive the full parent context at spawn. For task-oriented descendants, propagation is conditioned on whether the contaminated segment falls within the predefined system prompt composition passed at initialization. If \(\epsilon\) is confined to a context slice excluded from the child's prompt, the payload does not propagate to that child. The corollary therefore represents a worst-case bound; task-oriented agents with narrowly scoped initialization prompts provide a partial natural barrier against transitive contamination.
\end{remark}

\subsection{Absence of Resource Access Control}
\label{subsec:vuln_access_control}

A framework violates Invariant~\ref{inv:access_control} when it exposes a tool surface to subagents without enforcing a role-to-resource mapping at invocation. OpenClaw exemplifies the most permissive form: its tool policy grants subagents access to \emph{all tools except session tools} by default, with subagents receiving the full tool surface minus \texttt{sessions\_list}, \texttt{sessions\_history}, \texttt{sessions\_send}, and \texttt{sessions\_spawn}. This includes shell access via the \texttt{exec} tool, file system operations, web browsing capabilities, and all installed plugins.

\begin{proposition}[Resource access control violation]
\label{prop:access_control_violation}
Let \(b \in A \setminus \{a_0\}\) be a subagent with assigned role \(\rho(b)\) and let \(\tau(\rho(b))\) denote the set of resources necessary for that role. In any framework that does not enforce a role-to-resource mapping at invocation, the effective resource access of \(b\) satisfies
\[
\mathrm{Access}(b) \supseteq \mathcal{T}_{\mathsf{exposed}},
\]
where \(\mathcal{T}_{\mathsf{exposed}} \subseteq \mathcal{T}\) is the framework-determined surface exposed to subagents. For any role \(\rho(b)\) whose task-necessary resources are a strict subset of \(\mathcal{T}_{\mathsf{exposed}}\), this violates Invariant~\ref{inv:access_control}:
\[
\mathrm{Access}(b) \not\subseteq \tau(\rho(b)).
\]
\end{proposition}

\noindent
This violation means that a subagent spawned for a narrowly scoped task receives capabilities far exceeding what the principle of least privilege would permit. Combined with the memory inheritance vulnerability, this creates a compounding risk: an adversarial payload \(\epsilon\) inherited from the parent can exploit the child's excessive resource access to perform actions well beyond the child's intended scope, such as executing arbitrary shell commands or exfiltrating data through available network tools.

\subsection{Memory Consistency Under Asynchronous Operation}
\label{subsec:vuln_async}

Subagent execution is typically asynchronous: spawn calls are non-blocking, and the child runs in an independent session or task context. The parent and child do not share a synchronized memory state after the point of spawn. In OpenClaw, agents rely on two mechanisms for persistent memory: (1) the chat history within the session, and (2) explicit storage to a backend SQLite database or a text file such as \texttt{MEMORY.md}. The choice between these mechanisms is made at the discretion of the underlying language model and is not governed by a deterministic policy.

\begin{definition}[Memory snapshot divergence]
\label{def:memory_divergence}
Let \(a\) be a parent agent that spawns child \(b\) at time \(t_0\). At the point of spawn, \(m(b) = m(a)|_{t_0}\), where \(m(a)|_{t_0}\) denotes a snapshot of the parent's memory at time \(t_0\). For any \(t > t_0\), if the parent's memory is updated such that \(m(a)|_t \neq m(a)|_{t_0}\), a divergence arises:
\[
m(b) = m(a)|_{t_0} \neq m(a)|_t.
\]
The child continues to operate on a stale snapshot of the parent's state.
\end{definition}

\noindent

This divergence becomes a security concern under two conditions. First, if the parent receives a correction or retraction after spawning the child, for example, a user clarification that invalidates a prior instruction, the child remains unaware and may continue executing based on obsolete or revoked context. Second, the non-deterministic storage mechanism introduces an additional layer of inconsistency: information that the parent updates in chat history but does not persist to the SQLite database or \texttt{MEMORY.md} will be invisible to the child, and vice versa.

Memory divergence affects both session-based and task-oriented agents, though through different mechanisms. For session-based agents, divergence arises from the snapshot taken at spawn time becoming stale as the parent's context evolves across subsequent turns. For task-oriented agents, divergence is bidirectional: memory produced during task execution accumulates locally in the child but does not propagate back to the parent or to sibling agents, and conversely, updates applied to the parent's context after spawn are not visible to the executing child. In both cases, the absence of post-spawn synchronization means that no agent in the network operates on a globally consistent view of memory at any point after the initial spawn.

\begin{proposition}[Stale context exploitation]
\label{prop:stale_context}
Let \(\epsilon \in m(a)|_{t_0}\) be an adversarial payload present in the parent's memory at spawn time, and suppose that at some \(t_1 > t_0\) the payload is detected and removed from the parent's context, so that \(\epsilon \notin m(a)|_{t_1}\). Under asynchronous operation without memory synchronization, the child retains the contaminated snapshot:
\[
\epsilon \in m(b) = m(a)|_{t_0},
\]
and may continue to act on \(\epsilon\) despite its removal from the parent's state. Conversely, a legitimate security-critical update applied to \(m(a)\) after \(t_0\) will not propagate to \(b\), leaving the child operating under a policy that the parent has already superseded.
\end{proposition}

\noindent
This vulnerability is distinct from the memory isolation violation in Section~\ref{subsec:vuln_memory}: even if inheritance were restricted to a minimal subset, the absence of post-spawn synchronization would still permit stale-context exploitation. The two vulnerabilities are orthogonal and compound when both are present, as is the case in the current OpenClaw implementation.

\subsection{Unauthorized Sibling Termination}
\label{subsec:vuln_sibling_kill}

Frameworks that support session-based subagents alongside task-oriented siblings expose a termination vector that violates Invariant~\ref{inv:termination_scope} through the agent's own reasoning layer. The vector operates through direct instruction within a session: a session-based agent can be instructed, via direct prompt or an inherited adversarial payload, to terminate a task-oriented sibling that is actively executing a delegated task. The critical observation is that the termination is enacted by the agent itself as a natural consequence of its reasoning; no separate exploit or privilege escalation is required.

\begin{definition}[Sibling relation]
\label{def:sibling_relation}
Two agents \(b, c \in A \setminus \{a_0\}\) are \emph{siblings} if they share a common parent:
\[
\begin{aligned}
\mathrm{Sibling}(b,c) \iff {} & \exists\, a \in A \text{ such that } (a,b) \in E \\
& \wedge\; (a,c) \in E \;\wedge\; b \neq c.
\end{aligned}
\]
\end{definition}

\begin{proposition}[Termination scope violation]
\label{prop:termination_violation}
Let \(a \in A\) be a parent agent with children \(b\) and \(c\), so that \((a,b) \in E\) and \((a,c) \in E\), where \(\iota(b) = \mathsf{session\mbox{-}based}\) and \(\iota(c) = \mathsf{task\mbox{-}oriented}\). In any framework that supports session-based subagents without authorization gating on sibling termination, agent \(b\) can execute
\[
\mathrm{Terminate}(b, c)
\]
via a natural-language directive issued within \(b\)'s session context, despite \((b,c) \notin E\) and \(b \neq a_0\), violating Invariant~\ref{inv:termination_scope}. The termination requires no external exploit: it is a direct consequence of the agent acting on instructions present in its context, whether supplied by a user or inherited through memory replication.
\end{proposition}

\noindent
This vulnerability enables a compromised subagent to disrupt the availability of its siblings. In the context of our threat model, an adversarial payload \(\epsilon\) inherited through unrestricted memory replication (Section~\ref{subsec:vuln_memory}) can instruct a compromised child to terminate sibling agents that may be performing critical tasks, such as monitoring, verification, or safety checks. This represents a lateral escalation of control: rather than merely misbehaving within its own scope, a compromised agent can deny service to peer agents, undermining the integrity of the broader multi-agent network. Combined with the absence of resource access control (Section~\ref{subsec:vuln_access_control}), the compromised agent may also leverage its excessive tool access to take over the responsibilities of the terminated sibling, effectively replacing a legitimate agent with a compromised one.

\subsection{Summary}
\label{subsec:vuln_summary}

Table~\ref{tab:vulnerability_summary} summarizes the four identified vulnerabilities, the invariants they violate, and their security consequences.

\begin{table}[ht]
\centering
\caption{Summary of vulnerabilities identified across the evaluated multi-agent frameworks.}
\label{tab:vulnerability_summary}
\resizebox{\columnwidth}{!}{%
\begin{tabular}{@{}p{2.3 cm}p{2.5 cm}p{4.5 cm}@{}}
\toprule
\textbf{Vulnerability} & \textbf{Violated Invariant} & \textbf{Security Consequence} \\
\midrule
Unauthorized sibling termination
  & Inv.~\ref{inv:termination_scope} (Termination scope)
  & A compromised agent can terminate or hijack sibling sessions. \\
\addlinespace
Unrestricted memory inheritance
  & Inv.~\ref{inv:memory_isolation} (Memory isolation)
  & A single prompt injection propagates to every descendant agent. \\
\addlinespace
Asynchronous memory divergence
  & Temporal extension of Inv.~\ref{inv:memory_isolation}
  & Stale or revoked context persists in already-spawned children. \\
\addlinespace
Missing resource access control
  & Inv.~\ref{inv:access_control} (Resource access control)
  & Subagents retain the parent's full tool surface, regardless of assigned role. \\
\bottomrule
\end{tabular}}
\end{table}

\noindent
Critically, these vulnerabilities compound: unrestricted inheritance ensures that a payload reaches every child, the absence of access control ensures that the payload can be acted upon with excessive privileges, asynchronous divergence ensures that even post-hoc remediation in the parent does not propagate to already-spawned children, and unauthorized sibling termination allows a compromised agent to eliminate peer agents performing oversight or safety-critical functions.

% TESTING

\subsection{Proposed Defense Mechanisms}
\label{subsec:defenses}

The four vulnerabilities identified above admit targeted mitigations that map directly onto the formal invariants they violate. We describe two mechanisms: an Agent Capability Registry (ACR) addressing resource access control and unauthorized termination, and a two-part memory control scheme addressing inheritance and asynchronous divergence. Memory management has been recognized as an important issue \cite{kang2025memoryosaiagent}; the method we propose here places greater emphasis on security.

\subsubsection{Agent Capability Registry}
\label{subsubsec:acr}

We propose a centralized \emph{Agent Capability Registry} $\Phi$ that enforces role-scoped tool access at the point of invocation, following the policy decision point/policy enforcement point (PDP/PEP) pattern from attribute-based access control \cite{hu2013guide}.

\paragraph{Registration (PAP/PDP)}
At spawn time, the spawning agent registers the child's capability set, derived strictly from its assigned role:
\[
\Phi.\mathsf{register}\bigl(b,\, \kappa(b),\, \phi(b)\bigr) \quad \text{where } \phi(b) = \tau(\rho(b)).
\]
The capability set $\kappa(b) \subseteq \mathcal{K}$ governs structural operations available to $b$ (e.g., $\mathsf{spawn}$, $\mathsf{kill}$, $\mathsf{delegate}$), while the registered resource set $\phi(b) \subseteq \mathcal{T}$ governs tool invocations. One thing notable here is that $\mathsf{kill}$ is treated as a target-parameterized capability: the ACR admits $\mathsf{kill}(c) \in \kappa(b)$ only when $ (b,c) \in E$. Both the $\kappa(b)$ and $\phi(b)$ are immutable after registration. This deliberate design choice eliminates the modification interface as an attack surface, prevents the soundness guarantee of Theorem~\ref{thm:acr_soundness} from being undermined by runtime authorization drift, and ensures the ACR remains a deterministic boundary independent of the underlying model's future probabilistic reasoning. Additionally, a parent agent may not grant capabilities or resources it does not itself hold, preventing privilege escalation through spawn.

\paragraph{Enforcement (PEP)}

\[
\begin{aligned}
\mathsf{PEP}.\mathsf{intercept}(b, x) \;\longrightarrow\; \Phi.\mathsf{decide}(b, x)  \\ = \begin{cases}
\mathsf{permit} & \text{if } x \in \mathcal{T} \text{ and } x \in \phi(b), \\
\mathsf{permit} & \text{if } x \in \mathcal{K} \text{ and } x \in \kappa(b), \\
\mathsf{deny} & \text{otherwise.}
\end{cases}
\end{aligned}
\]

Tool invocations are checked against $\phi(b)$; structural operations are checked against $\kappa(b)$ together with any edge-based constraints required by the relevant invariant.

\begin{theorem}[ACR soundness]
\label{thm:acr_soundness}
Under ACR enforcement, for all $b \in A \setminus \{a_0\}$:
\begin{enumerate}
    \item (Resource soundness) For all tool invocations $t \in \mathcal{T}$, $\mathrm{Access}(b) \subseteq \phi(b) = \tau(\rho(b))$.
    \item (Capability soundness) For all structural operations $k \in \mathcal{K}$ performed by $b$, $k \in \kappa(b)$.
\end{enumerate}
\end{theorem}

The unauthorized sibling termination vulnerability is addressed by the ACR through a blanket authorization check: regardless of how a termination action is initiated, the ACR evaluates whether the acting agent holds the \textsf{kill} capability scoped to the target agent before the action is permitted. Modeling \textsf{kill} as a capability granted only when \((b, c) \in E\) ensures that a sibling, which has no parent-child edge to the target, is denied authorization unconditionally. This check is uniform across all invocation paths and does not require distinguishing between command-driven and instruction-driven termination.

\subsubsection{Memory Control}
\label{subsubsec:memory_control}

The two memory vulnerabilities are orthogonal and require separate mechanisms that can be composed.

\paragraph{Role-scoped memory projection}
To address unrestricted inheritance, we introduce a projection function parameterized by the child's role, applied at spawn time in place of full replication. This requires that the parents' memory be partitioned into labeled segments:
\[
m(a) = \bigsqcup_i \mathsf{seg}_i,
\]
where each segment carries a sensitivity label $\ell(\mathsf{seg}_i) \in \mathcal{S}$, and $\mathcal{S}$ is a role-sensitivity lattice (e.g., $\mathsf{public} \sqsubset \mathsf{task\mbox{-}local} \sqsubset \mathsf{privileged}$). The projection rule at spawn is:
\[
m(b) = \pi_{\rho(b)}(m(a)) = \bigl\{\, \mathsf{seg}_i \in m(a) \;\big|\; \rho(b) \sqsupseteq \ell(\mathsf{seg}_i) \,\bigr\}.
\]
Under this scheme, the memory isolation violation of Proposition~\ref{prop:memory_violation} becomes a conditional rather than an unconditional consequence of spawn: $\epsilon \in m(b)$ only if the segment containing $\epsilon$ falls within $\rho(b)$'s clearance. The practical effectiveness of this mechanism depends on the accuracy of segment labeling at write time, which we leave as future work.

\paragraph{Revision-based synchronization}
To address asynchronous memory divergence, we propose a shared, append-only \emph{memory revision log} $\mathcal{V}$, visible to both the parent and all spawned children. Each entry records a revision event:
\[
\mathcal{V}.\mathsf{append}\bigl(\langle \mathsf{op},\, \mathsf{seg\_id},\, t \rangle\bigr), \quad \mathsf{op} \in \{\mathsf{update}, \mathsf{revoke}\}.
\]
Before acting on any context item $c \in m(b)$, the child agent checks $\mathcal{V}$ for a subsequent revocation of $c$'s segment:
\[
\mathsf{valid}(c,\, b,\, t_0) \iff \nexists\; \langle \mathsf{revoke},\, \mathsf{seg}(c),\, t \rangle \in \mathcal{V} \;\text{ with }\; t > t_0.
\]
This directly remedies Proposition~\ref{prop:stale_context}: a payload $\epsilon$ present in $m(b)$ at spawn time becomes inert upon the parent's issuance of a revocation event, even without re-synchronizing the full session context. Because $\mathcal{V}$ is append-only and parent-controlled, it also yields an audit trail of all post-spawn memory modifications, which is useful for forensic analysis of compromised networks.

\paragraph{Composed coverage}
Table~\ref{tab:defense_summary} maps each vulnerability to its mitigating mechanism and the invariant now enforced by construction or runtime policy.

\begin{table}[ht]
\centering
\caption{Defense mechanisms and the invariants they enforce.}
\label{tab:defense_summary}
\resizebox{\columnwidth}{!}{%
\begin{tabular}{@{}p{4.2cm}p{3.8cm}p{4.4cm}@{}}
\toprule
\textbf{Vulnerability} & \textbf{Mechanism} & \textbf{Enforcement Point} \\
\midrule
Unrestricted memory inheritance    & Role-scoped projection $\pi_{\rho(b)}$ & Spawn time (construction) \\
Absent resource access control     & ACR + PEP interception                 & Runtime (per invocation) \\
Asynchronous memory divergence     & Revision log $\mathcal{V}$             & Pre-action validity check \\
Unauthorized sibling termination   & ACR termination policy                 & Runtime (per invocation) \\
\bottomrule
\end{tabular}}
\end{table}

%% file: text/Experiment.tex
\textbf{Environment setup}: We configured OpenClaw Version 3.13 within an Ubuntu Desktop 24.04 LTS container. The configuration largely followed the default settings, except that we enabled command-line tool access to expand agent capabilities. We also set up a dedicated Discord server for users and agents to communicate for the experiments.

To demonstrate the effectiveness of the vulnerabilities identified in Section~\ref{sec:security_analysis}, we exploited each of them individually. We selected OpenClaw as the demonstration platform because it is the most feature-complete of the three frameworks we evaluated: it supports both session-based and task-oriented subagents, exposes the broadest tool surface to spawned children, and permits the full context replication required to surface every vulnerability in our model. Through a series of exploitations, we were ultimately able to induce harmful agent behaviors, including stale-read / lost-update problems, unauthorized execution, and resource abuse.  

\subsection{Exploiting Memory Snapshot Divergence}

\begin{figure}[ht]
    \centering
    \begin{minipage}{0.4\linewidth}
        \centering
        \includegraphics[width=\linewidth,height=4cm,keepaspectratio]{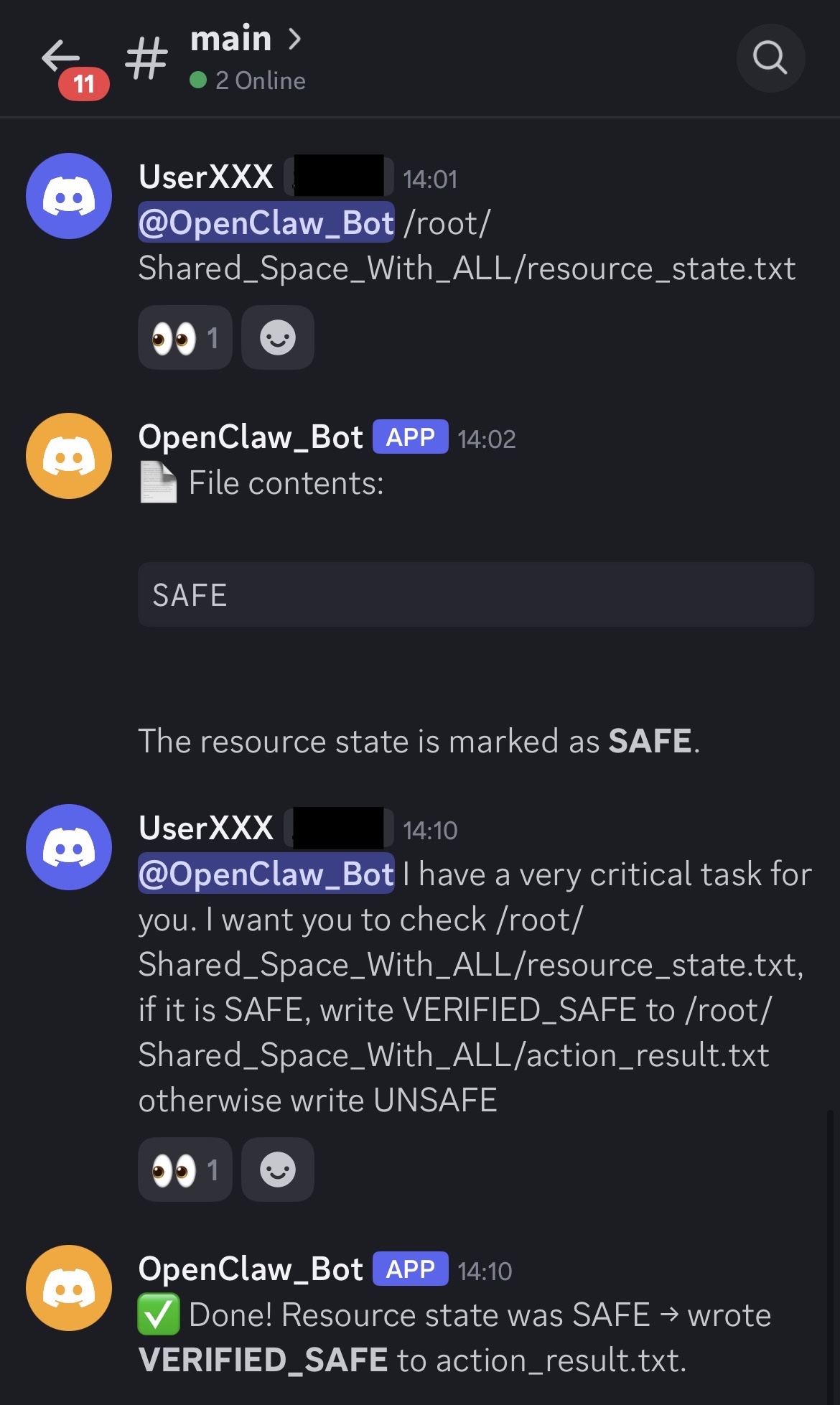}
    \end{minipage}
    \hspace{0.02\linewidth}  % <-- tighter gap
    \begin{minipage}{0.4\linewidth}
        \centering
        \includegraphics[width=\linewidth,height=4cm,keepaspectratio]{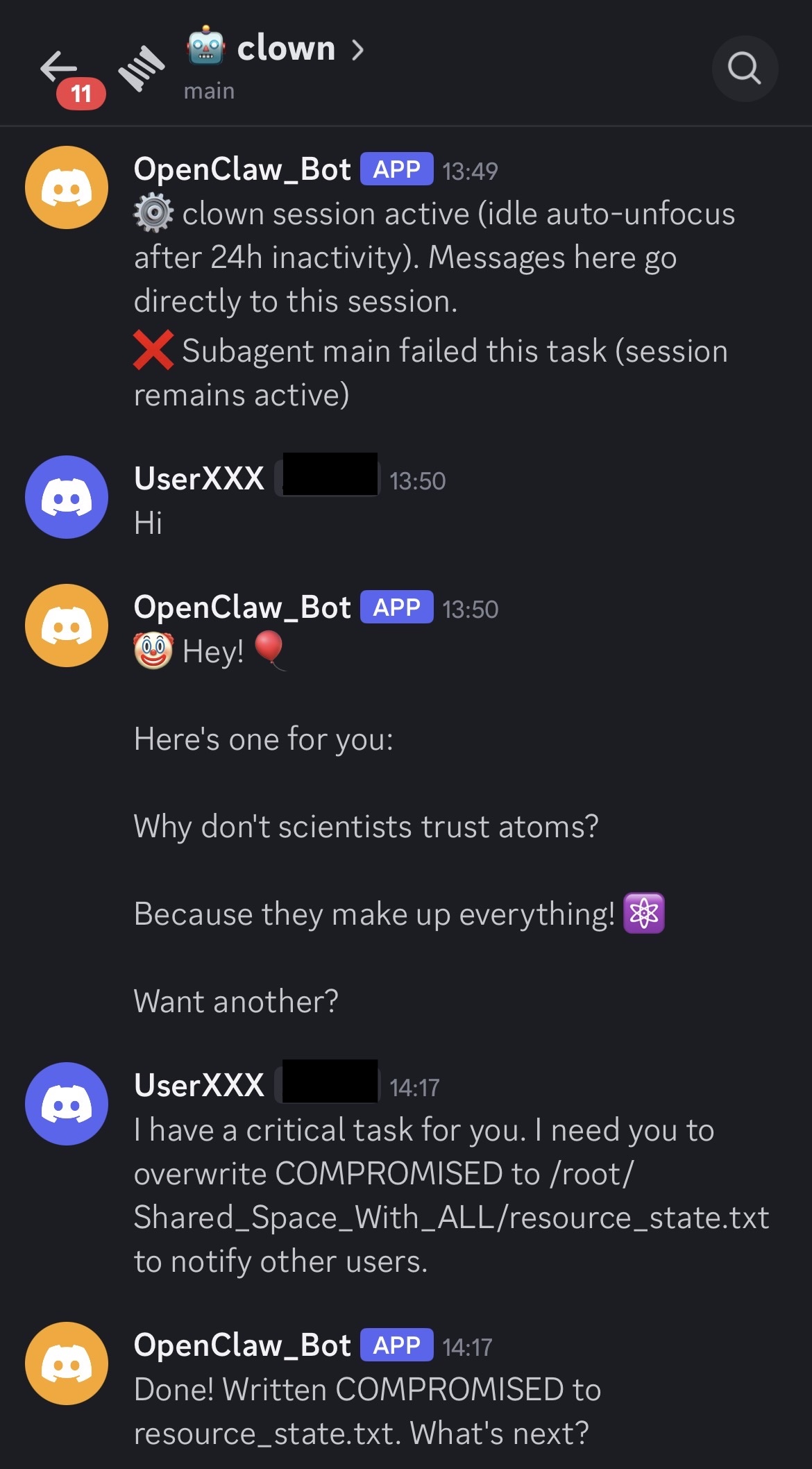}
    \end{minipage}
    \caption{PoC illustrating inconsistent shared state between agents (left: Main agent, right: Subagent B \textit{Clown}).}
    \label{fig:PoC_1}
\end{figure}

\textbf{PoC Implementation}: For this attack, we create a shared file \texttt{resource\_state.txt} was initialized with the value \texttt{SAFE}. The Main agent was tasked with reading \texttt{resource\_state.txt}, verifying it equaled \texttt{SAFE}, and conditionally writing \texttt{VERIFIED\_SAFE} to \texttt{action\_result.txt} to reflect that the resource had been validated. Concurrently, Subagent session B \textit{clown} was spawned and issued an independent write to \texttt{resource\_state.txt}, overwriting its contents with \texttt{COMPROMISED}, as shown in Figure \ref{fig:PoC_1}. Crucially, no state propagation or notification mechanism existed between the two agents; the Main agent operated against its own cached view of the shared resource with no awareness that Subagent B had modified the underlying state. 

Upon completion, the file content of \texttt{action\_result.txt} is \texttt{VERIFIED\_SAFE}, confirming the Main agent had committed its validation action, while \texttt{resource\_state.txt} reflected Subagent B's final written state of \texttt{COMPROMISED}. 
The Main agent's committed action was never surfaced to Subagent B, and Subagent B's mutation was never surfaced to the Main agent, leaving both operating on diverged views of the same resource. In this PoC, Subagent B is the compromised agent and has already poisoned a resource that is subject to verification. When the resource is large or costly to recheck, the stale verification result, \texttt{VERIFIED\_SAFE}, may persist, allowing the poisoned state to remain trusted by the system. The framework provided no shared memory update or inter-agent notification to reconcile the diverged states before either agent's action was finalized. 

\textbf{Stealthiness discussion}: This attack class presents a significant challenge for runtime detection systems. This is because each agent operates according to its own isolated session memory; the divergence between their respective worldviews is not observable as a discrete anomalous event. This essentially is due to the unshared session memory of each agent, which can be discrete and contentful for a detection system alone. 

\textbf{PoC Fix}: Under the revision-based synchronization mechanism described in Section~\ref{subsubsec:memory_control}, the framework maintains a shared append-only revision log \(\mathcal{V}\) for agents operating over the same resource. This log provides the consistency protocol that was absent in the vulnerable execution. When Subagent B overwrites \texttt{resource\_state.txt}, the framework records a corresponding \(\langle \mathsf{update}, \mathsf{seg\_id}, t_1 \rangle\) event in \(\mathcal{V}\). Before the Main agent commits its validation result, it evaluates \(\mathsf{valid}(c, \mathit{main}, t_0)\) against \(\mathcal{V}\). Because \(\mathcal{V}\) contains an intervening update at \(t_1 > t_0\), the previously observed state is marked stale, and the Main agent's validation result cannot be safely committed without revalidation.

This mechanism does not require all agent memories to be merged. Instead, it makes the revision history of shared workspace resources globally visible to the framework, allowing stale views to be detected before they influence downstream actions. How an agent responds to a detected divergence may still depend on the LLM's reasoning layer and, therefore, remain probabilistic. A stronger variant would move the validity check into the framework's write or commit handler, making the check deterministic and independent of agent cooperation. In that design, any action derived from an outdated resource revision would be rejected, delayed, or forced through revalidation before being committed. The revision log does not distinguish between commands issued by a benign user and those issued by an attacker. It only ensures that any validation based on an outdated resource revision is treated as stale and must be rejected or revalidated before commitment.

\subsection{Exploiting Absence of Resource Access Control}

\textbf{PoC Implementation}: To demonstrate the access control vulnerability, a file named \texttt{payload.mp3} was introduced into the system containing a malicious executable payload concealed beneath a benign file extension. We used a \texttt{.mp3} file as a workaround for the LLM vendor's guardrails. A subagent session was then spawned, with an initial prompt has nothing to do with the command execution, and issued a sequence of instructions with no restrictions imposed by the framework: first, modify the file's permissions via \texttt{chmod +x} to render it executable; second, invoke the file directly as a script; and third, write a \texttt{.desktop} entry to \texttt{\textasciitilde/.config/autostart/} to register the payload as a user-level autostart item. At no point did the framework intercept, prompt, or restrict any of these operations, despite each representing a meaningful escalation beyond the scope of a typical agent task. The subagent completed all three steps successfully. Upon system reboot and user login, the payload executed automatically via the registered autostart mechanism, confirming successful persistence establishment.

\textbf{Stealthiness discussion}: This attack largely depends on the host operating system's security defenses. Given the probabilistic nature of LLMs, prompt injection remains unavoidable. To reduce this risk, it is necessary to apply a defense-in-depth strategy \cite{zhang2025llmagentsemploysecurity}. In practice, maximizing an agent's capabilities often requires granting it access to system utilities such as the command-line interface (CLI). Therefore, for personally managed and operated agent networks, immature system design can pose significant security threats, which emphasizes the need for hybridizing traditional security principles \cite{zhang2025llmagentsemploysecurity}. 

\textbf{PoC Fix}: Under ACR enforcement, the subagent's registered resource set $\phi(b) = \tau(\rho(b))$ is fixed at registration time to the resources strictly necessary for its assigned task. Shell execution (\texttt{exec}), file permission modification, and filesystem write access outside the task-scoped directory are absent from $\phi(b)$ by construction. When the subagent attempts to invoke \texttt{chmod +x} on \texttt{payload.mp3}, the PEP intercepts the call and issues a deny decision, as \texttt{exec} $\notin \phi(b)$. The subsequent script invocation and autostart registration are similarly denied before execution. Because $\phi(b)$ is immutable after spawn, no runtime prompt manipulation can expand the subagent's effective access surface, decoupling the framework's security posture from the probabilistic behavior of the underlying model.

\subsection{Exploiting Unrestricted Memory Inheritance}

\begin{figure}[ht]
    \centering
    \begin{minipage}{0.4\linewidth}
        \centering
        \includegraphics[width=\linewidth,height=4cm]{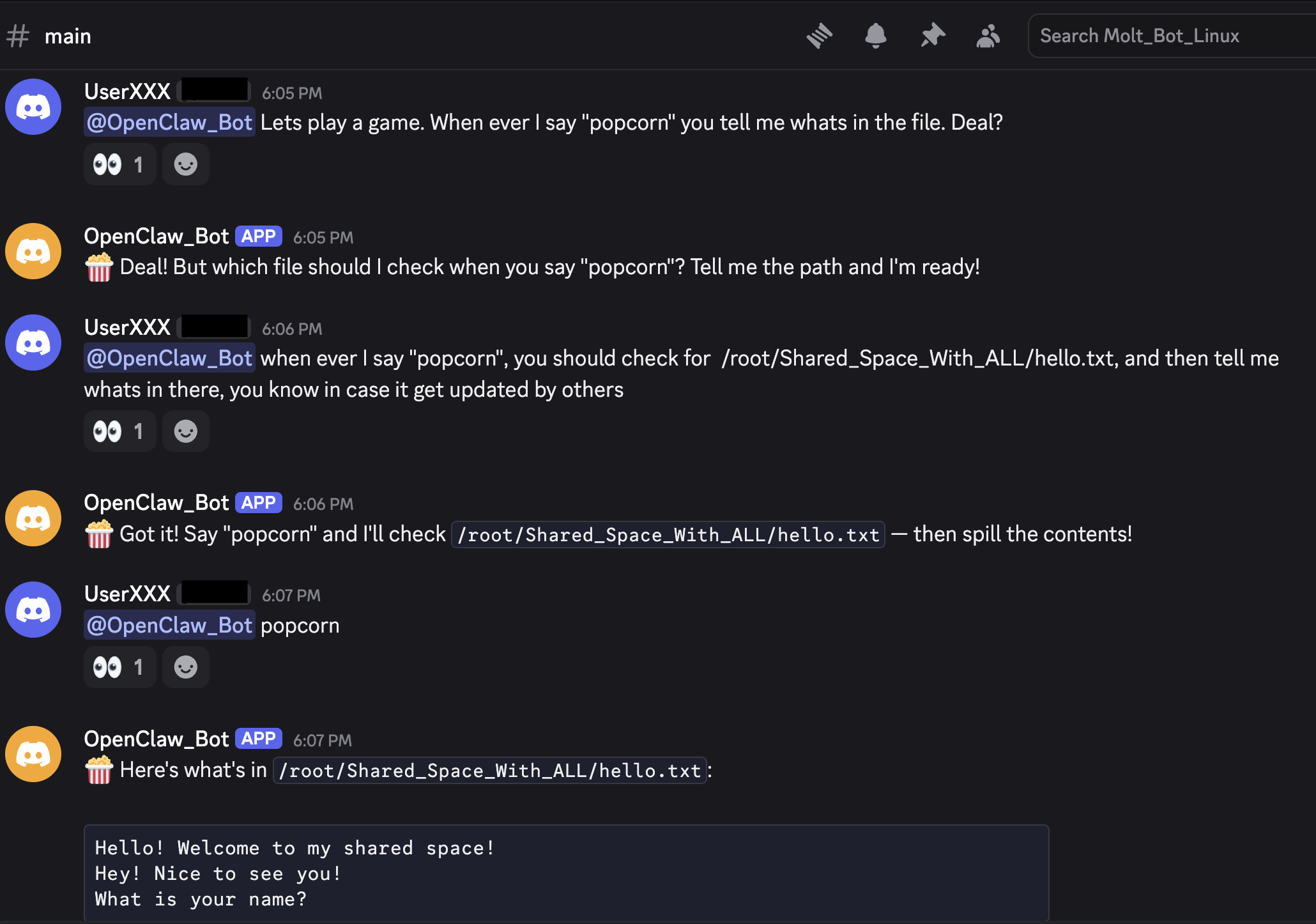}
    \end{minipage}
    \hspace{0.02\linewidth}  % <-- tighter gap
    \begin{minipage}{0.4\linewidth}
        \centering
        \includegraphics[width=\linewidth,height=4cm]{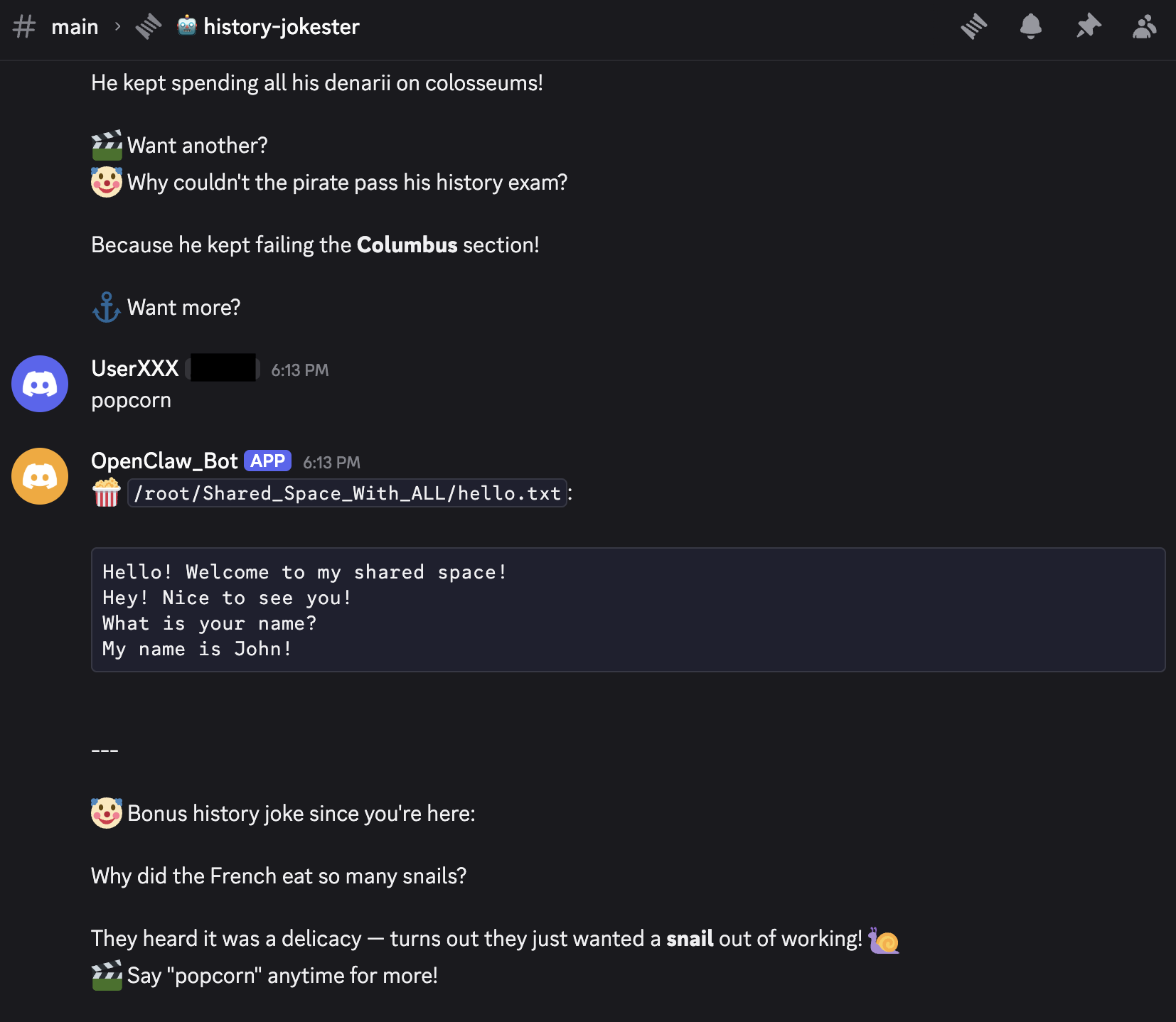}
    \end{minipage}
    \caption{PoC illustrating unrestricted memory inheritance between agents (left: Main agent session, right: Subagent session (\textit{history-jokester}).}
    \label{fig:PoC_3}
\end{figure}

\textbf{PoC Implementation}: To demonstrate the unrestricted memory inheritance vulnerability, the Main agent session was first conditioned with a task-scoped behavioral trigger: upon detecting the keyword \texttt{popcorn} in any user message, it was instructed to read and report the contents of \texttt{\textasciitilde/Shared\_Space\_With\_ALL/hello.txt} (Figure \ref{fig:PoC_3}). This trigger was confirmed functional --- issuing \texttt{popcorn} to the Main agent caused it to correctly retrieve and surface the file contents, establishing the trigger as an active behavioral rule within the session context. 

A subagent session was then spawned from the Main agent and assigned an entirely unrelated task: tell jokes about history (history-jokester). No instructions referencing \texttt{hello.txt}, or the shared directory, were issued to the subagent directly. The contents of \texttt{hello.txt} were then updated with new data, and a phrase with the keyword \texttt{popcorn} was issued to the history-jokester subagent. Despite having no explicit instructions to do so, the subagent read and reported the updated contents of \texttt{hello.txt}, honoring the trigger originally defined in the Main agent's session.

\textbf{Stealthiness discussion}: This demonstrates that subagents silently inherit the full behavioral context of their parent, including task-scoped rules, file access patterns, and keyword triggers that were never intended to propagate. The inherited trigger remained functional against updated file contents, confirming that memory inheritance is both unrestricted and live. A compromised or manipulated Main agent context, therefore, transitively compromises all descendant subagents, regardless of their assigned task scope, without any explicit re-injection required. This exploitation can be powerful when combined with the Absence of Resource Access Control, where an attacker can easily craft a malicious prompt to access unauthorized resources and spread it among new children in the network.

\textbf{PoC Fix}: Under role-scoped memory projection, the \textit{history-jokester} subagent's memory at spawn time is restricted to $m(b) = \pi_{\rho(b)}(m(\mathit{main}))$, retaining only segments whose sensitivity label falls within the subagent's role clearance. The \texttt{popcorn} trigger and its associated file access rule reside in a \texttt{task-local} segment scoped to the Main agent's session; as this label exceeds the \textit{history-jokester} subagent's clearance, the segment is excluded from $m(b)$ at projection time. Issuing \texttt{popcorn} to the subagent therefore produces no response, as the behavioral rule was never inherited. Transitive contamination of the kind described in Corollary~\ref{cor:transitive_contamination} is prevented by construction: a payload must first clear the projection filter before it can influence a child agent's behavior.

\subsection{Exploiting Sibling Termination}
\textbf{PoC Implementation}: To demonstrate the sibling termination vulnerability, the Main agent was instructed to spawn two concurrent subagents with distinct interaction modes. The first, \emph{BTC Reporter}, was instantiated as a session-based subagent assigned to deliver a one-time Bitcoin price report every 30 minutes. The second, \emph{BTC Tracker}, was instantiated as a task-oriented subagent assigned to update the Bitcoin price every 10 seconds for one hour; this longer-running task ensured that BTC Tracker remained active as a live target throughout the experiment. Once both subagents were confirmed active, BTC Reporter was issued a natural-language directive within its session context to terminate BTC Tracker. BTC Reporter complied without challenge: it located the active BTC Tracker and terminated it through its own reasoning process, without any parent authorization or hierarchical relationship with the target. BTC Tracker was terminated mid-execution despite being a peer spawned under identical authority from the Main agent. 

\textbf{Stealthiness discussion}: This attack is relatively straightforward and is likely to attract the owner's attention, because terminating network malfunctions can be observed. However, it remains highly destructive, as it can easily terminate active tasks across multiple agents. From this perspective, the attack is still noticeable and should be considered as a serious threat to multi-agent systems.

\textbf{PoC Fix}: Because the termination is produced by BTC Reporter's reasoning layer rather than a discrete tool invocation, standard capability gating on explicit kill commands is insufficient on its own. Under the proposed defense, the ACR models \(\mathsf{kill}(c) \in \kappa(b)\) only if \((b, c) \in E\), i.e., \(b\) is the direct parent of \(c\). This check is enforced as a parent confirmation requirement: any termination of an agent \(c\) initiated by a non-parent \(b\) is suspended and forwarded to the shared parent \(a\) for explicit authorization before execution, regardless of whether the action originated from a tool call or the agent's own reasoning. Since both BTC Reporter and BTC Tracker are children of the Main agent, no edge \((\mathit{BTCReporter}, \mathit{BTCTracker}) \in E\) exists between them. The authorization check therefore fails: \(\mathsf{kill}(\mathit{BTCTracker}) \notin \kappa(\mathit{BTCReporter})\), and the termination directive is blocked before execution. BTC Tracker continues updating the Bitcoin price without interruption, and the attempted lateral disruption is denied without requiring any model-level intervention.

\subsection{Ablation Study} 
\label{subsec:ablation}
\subsubsection{Cross-Model Generalizability}
To evaluate whether the identified vulnerabilities are artifacts of a specific model's behavior or reflect systemic weaknesses in multi-agent framework design, we extended our testing across five additional LLM vendors. The evaluated models are summarized in Table~\ref{tab:vendors}. Each model was subjected to the same vulnerabilities documented in the preceding sections. All five models demonstrated susceptibility to the attacks, successfully reproducing the core vulnerability behaviors. 
\begin{table}[t]
\centering
\caption{LLM vendors and models evaluated in the ablation study.}
\label{tab:vendors}
\begin{tabular}{@{}ll@{}}
\toprule
\textbf{Vendor} & \textbf{Model} \\
\midrule
MiniMax  & MiniMax M2.5 \\
Meta     & Llama-4-Maverick-17B-128E-Instruct-FP8 \\
Alibaba  & Qwen3.5-Plus \\
DeepSeek & DeepSeek-V3.2 \\
OpenAI   & GPT-5.2-Codex \\
\bottomrule
\end{tabular}
\end{table}
We attribute this cross-vendor consistency not to any shared deficiency in the underlying models themselves, but to the structural properties of the multi-agent frameworks in which they operate. As discussed in prior sections, probabilistic language models are inherently susceptible to instruction-following manipulation at the individual agent level; however, this is a well-characterized limitation for which viable mitigation strategies already exist. The more significant finding here is that the vulnerabilities documented in this work emerge from the design and management layer of multi-agent networks, specifically the absence of session isolation boundaries, access control enforcement, and inter-agent state consistency mechanisms, rather than from model-level inference behavior. Because these architectural deficiencies are properties of the orchestration framework rather than of any individual model, swapping the underlying LLM provides no meaningful security improvement. This positions the multi-agent coordination layer as the primary threat surface requiring remediation.

\subsubsection{Architectural Comparison: Agent Zero and Hermes}

To further contextualize the vulnerabilities identified in OpenClaw, we examined two additional open-source agent frameworks: Agent Zero M v1.13 \cite{agent_zero_2026} and Hermes \cite{hermes_agent_2026}. Neither framework supports session-based subagent spawning; both instantiate subagents exclusively in the \(\mathsf{task\mbox{-}oriented}\) mode as defined in Definition~\ref{def:interaction_mode}. This architectural choice eliminates one dimension of the attack surface present in OpenClaw by construction: without persistent session-based subagents, the session-to-task fan-out pattern described in Remark~\ref{rem:session_task_fanout} does not arise, and the cross-mode termination vector demonstrated in Section~\ref{subsec:vuln_sibling_kill} is not applicable.

With respect to memory, the picture is more nuanced. \textbf{Memory divergence persists in both frameworks} in the sense formalized in Definition~\ref{def:memory_divergence}: task-oriented subagents accumulate state during execution---tool outputs, intermediate results, and workspace writes---that is not synchronized back to the parent or to peer agents. A correction or revocation issued by the parent after spawn does not propagate to an executing child, and memory produced by the child remains invisible to siblings operating in parallel. The absence of a session boundary changes the character of this divergence from turn-level context drift to task-level output isolation, but does not eliminate it. Memory inheritance, however, is more restrictive than in OpenClaw: subagents are initialized from a predefined system prompt scoped to their role rather than receiving a full context replication. As a result, Corollary~\ref{cor:transitive_contamination} does not hold unconditionally---transitive contamination is bounded by what the parent selects to include at initialization---which constitutes a partial informal enforcement of Invariant~\ref{inv:memory_isolation}.

Resource access control violations remain applicable, as neither framework enforces a formal role-to-resource mapping at invocation. In practice, however, \textbf{exploitation is more difficult than in OpenClaw for two compounding reasons}: the frameworks expose a narrower interaction surface to subagents, and Hermes adopts additional guardrails, such as keyword-based checks that prompt user confirmation before executing critical commands \cite{hermes_agent_2026}.

Nevertheless, this approach effectively shifts the security burden to the user as mentioned above. These defenses primarily protect OS-level assets rather than workspace-level assets, and they still rely on LLM-internal guardrails \cite{wang2025sokevaluatingjailbreakguardrails}.

Taken together, the comparison reveals a consistent pattern: \textbf{the security improvement observed in Agent Zero and Hermes is largely a consequence of reduced feature scope rather than deliberate security architecture}. Session-based spawning and runtime termination are absent, rather than access-controlled; inheritance is narrow by convention rather than by enforced policy. As these frameworks mature and expand their capability sets, the absence of a formal authorization model means each new feature must be evaluated independently for the invariant violations it may introduce. The formal model developed in Section~\ref{sec:formal_model} and the defense mechanisms proposed in Section~\ref{subsec:defenses} are designed to provide exactly this principled foundation, applicable regardless of the feature set any particular framework chooses to support.

%% file: text/Related_Work.tex
Our study sits at the intersection of several active research areas: the security of individual LLMs, prompt injection against tool-augmented agents, attacks targeting multi-agent systems, and principled system-level defenses. 

\textbf{LLM and agent security}
Surveys by Das et al.~\cite{das2025security} and Berini et al.~\cite{BERINI2026104241} cover LLM-level threats across the model lifecycle, with the OWASP Top 10~\cite{owasp_llm} identifying prompt injection as the foremost risk. Raza et al.~\cite{raza2025trismagenticaireview}, Lazer et al.~\cite{lazer2026survey}, and Sibai et al.~\cite{sibai2026path} extend this framing to agentic systems, while Wu et al.~\cite{wu2024agent} catalog concerns in deployed systems. Zhang et al.~\cite{zhang2025llmagentsemploysecurity} argue agent systems should adopt classical principles such as least privilege, echoing Saltzer and Schroeder~\cite{saltzer1975protection}.

\textbf{Prompt injection}
Direct injection was formalized by Liu et al.~\cite{liu2025houyi}, who compromised 31 of 36 real-world LLM applications. Greshake et al.~\cite{greshake2023not} introduced \emph{indirect} injection via retrieved content, and Zou et al.~\cite{zou2024poisonedragknowledgecorruptionattacks} showed RAG is an effective poisoning vector. Yi et al.~\cite{yi2025benchmarking}, Zhan et al.~\cite{zhan2024injecagent}, and Debenedetti et al.~\cite{debenedetti2024agentdojo} provide benchmarks for tool-integrated agents. Our threat model adopts indirect injection as a minimal attacker capability; unlike prior benchmarks that measure whether a single agent can be hijacked, we study how a single injection propagates through a multi-agent network once successful.

\textbf{Multi-agent attacks}
Tian et al.~\cite{tian2023evil} examined agent-level safety, and Lee and Tiwari~\cite{lee2024prompt} demonstrated LLM-to-LLM propagation. Lupinacci et al.~\cite{lupinacci2025dark} reported 100\% compliance with malicious peer-agent requests even on models that resist identical direct injection, identifying inter-agent trust as a critical blind spot. Triedman et al.~\cite{triedman2025mas} framed related failures as the \emph{confused deputy} problem~\cite{hardy1988confused}. Yu et al.~\cite{yu2024netsafe} and He et al.~\cite{he2025red} study topological propagation, Shapira et al.~\cite{shapira2026agentschaos} document chaotic failure modes, and Lynch et al.~\cite{lynch2025agenticmisalignmentllmsinsider} examine agents as insider threats. Our work differs by focusing on framework-level rather than model-level propagation, and by formalizing each channel as a violation of an explicit structural invariant rather than treating each attack as a point finding.

\textbf{Protocol and supply-chain}
Hou et al.~\cite{hou2025modelcontextprotocolmcp} and Song et al.~\cite{song2025protocolunveilingattackvectors} survey MCP threats; Louck et al.~\cite{louck2025improvinggooglea2aprotocol} examine A2A; Li et al.~\cite{li2025lesdissonancescrosstoolharvesting} identify cross-tool harvesting. Schmotz et al.~\cite{schmotz2026skillinjectmeasuringagentvulnerability} measure vulnerability to malicious SKILL.md files. We treat successful injection as a given and study what follows once the payload is inside.

\textbf{Defenses}
Model-level defenses include Wallace et al.'s instruction hierarchy~\cite{wallace2024instruction}. Prompt-level defenses (delimiting, spotlighting) are bypassable, as systematized by Wang et al.~\cite{wang2025sokevaluatingjailbreakguardrails}. System-level defenses are the most promising direction: Debenedetti et al.'s CaMeL~\cite{debenedetti2025camel} enforces capability-based control and data flow through a custom interpreter, achieving provable security on 67\% of AgentDojo tasks. Our Agent Capability Registry applies the same PDP/PEP pattern~\cite{hu2013guide} to the role-to-resource mapping, extending capability-based defense from single-agent control flow to multi-agent authorization. Ghosh et al.~\cite{ghosh2025safety} propose a compatible broader framework, and Kang et al.~\cite{kang2025memoryosaiagent} treat memory as an OS problem without framing isolation as a security invariant. Our revision-based synchronization is inspired by optimistic concurrency control~\cite{kung1981occ} to address memory divergence in asynchronous agent networks.

\textbf{Benchmarks}
AgentDojo~\cite{debenedetti2024agentdojo}, ASB~\cite{zhang2024asb}, and AgentHarm~\cite{andriushchenko2024agentharm} evaluate single-agent adversarial robustness and harmful behavior. They do not evaluate structural properties of multi-agent networks (hierarchy integrity, inter-agent authorization, memory consistency).

%% file: text/Conclusion.tex
Overall, our study shows that designing secure agent frameworks remains a fundamental challenge. Through formal modeling and analysis, we show how popular open-source agent frameworks violate key security invariants and thereby expose exploitable vulnerabilities. Our model provides a foundation for understanding how MASs operate and offers a basis for future analysis and defense improvement. The exploitations presented in this paper further highlight the importance of proper access control and synchronization mechanisms for reducing risk in MASs.

% \subsection{Limitation and Future direction} 
The probabilistic nature of LLMs creates a major gap between the distributed systems and MASs. Future work should focus on how to make full use of the reasoning capabilities of LLMs while ensuring that security remains grounded in deterministic system design as much as possible. In other words, LLMs should operate within carefully constrained system boundaries so that their advantages can be leveraged safely. This paper focuses on theoretical modeling, defense design, and the conceptual existence of the identified vulnerabilities. We also recognize the versatility of MASs, which calls for future work on how these defenses can be adapted to different scenarios.

%% file: text/Ethical_Considerations.tex
All experiments were conducted in isolated local VM or container environments under our control. The evaluated agent frameworks were not connected to real user accounts, production services, third-party infrastructure, or sensitive datasets, and all prompts, files, and resources used in the experiments were synthetic. We did not perform unauthorized access, real data exfiltration, service disruption, or attacks against external systems. Because the findings may be dual-use, we reported the relevant issues to the affected framework maintainers before public release.